\documentclass[11pt]{amsart}
\usepackage{graphicx,epsfig}
\usepackage{multirow}
\usepackage{algorithm}
\usepackage{algpseudocode}

\usepackage{amsmath,amssymb,latexsym, amsfonts, amscd, amsthm, xy}
\input{xy}
\xyoption{all}

\input{xy}
\xyoption{all}

\numberwithin{equation}{section}

\usepackage[usenames]{color}

\usepackage{parskip}

\makeindex
=651pt \headheight = 12pt \headsep = 0pt

\theoremstyle{plain}
\newtheorem{theorem}{Theorem}[section]

\newtheorem{corollary}[theorem]{Corollary}

\newtheorem{lemma}{Lemma}[section]

\theoremstyle{definition}

\newtheorem{remark}{Remark}[section]

\DeclareMathOperator{\argmin}{argmin}

\begin{document}

\title{Bayesian Monotone Regression using Gaussian Process Projection}

\author{Lizhen Lin and David B. Dunson}
\email{lizhen@stat.duke.edu}
\email{dunson@stat.duke.edu}

\address{Department of Statistical Science, Duke University, Durham, NC 27708-0251, USA.  }

\maketitle

\begin{abstract}
Shape constrained regression analysis has applications in dose-response modeling, environmental risk assessment, disease screening and many other areas. Incorporating the shape constraints can improve estimation efficiency and avoid implausible results. We propose two novel methods focusing on Bayesian monotone curve and surface estimation using Gaussian process  projections. The first projects samples from an unconstrained prior, while the second projects samples from the Gaussian process posterior. Theory is developed on continuity of the projection, posterior consistency and rates of contraction. The second approach is shown to have an empirical Bayes justification and to lead to simple computation with good performance in finite samples. Our projection approach can be applied in other constrained function estimation problems including in multivariate settings.

\end{abstract}

%\begin{keywords}
%Asymptotics; Bayesian nonparametrics; Isotonic regression; Projective Gaussian process; Shape constraint.
%\end{keywords}

\section{Introduction}

In a rich variety of applications, prior knowledge is available on the shape of a surface, with examples including monotonicity, unimodality and convexity. Incorporating such shape constraints can often substantially improve estimation efficiency and stability, while producing results consistent with prior knowledge.  We propose two novel approaches based on Gaussian process projections.  Gaussian processes are routinely applied but have the disadvantage of not allowing constraints.  Although we focus on monotone curves and surfaces, the approach can be applied directly in much broader settings including additive models, multivariate regression with monotonicity constraints only in certain directions, and other types of shape constraints.

    There is a rich frequentist literature on monotone curve and isotonic regression estimation, with a common approach minimizing a least squares loss subject to a restriction (Barlow et al. 1972, Robertson et al., 1988). For more recent references,  refer to \cite{Bhattacharya1} and Bhattacharya $\&$ Lin (2010, 2011). Alternatively, restricted kernel (\cite{Muller2}, Dette et al. (2005) and Mammen (1991)) and spline (Ramsay (1988) and  \cite{kong2}) methods have been proposed.

    From a Bayesian perspective,  one specifies a prior on the regression function and inference is  based on the posterior distribution.  \cite{alan91} use an ordered Dirichlet prior on a strictly monotone dose-response function.  \cite{neelon} use  an additive model with a prior imposed on the slope of the piecewise linear functions.  Shively et al. (2009) and Shively et al. (2011) use restricted splines.  \cite{Bornkamp09}  adopt mixture modeling.

    Although there is a rich existing literature on Bayes monotone curve estimation, our work has two key motivations: (1) There is a lack of theory supporting these methods beyond consistency; (2) Current approaches involve basis expansions and challenges arise in multivariate cases. Gaussian processes have a rich theoretical foundation, can easily incorporate prior information, and can be implemented routinely in multivariate settings. We define a class of projective Gaussian processes which inherit these advantages.

\section{Gaussian process projections}

Let $w \sim  \mbox{\small{GP}}( \mu, R)$ denote the sample path of a `mother' Gaussian process indexed on $\mathcal{X} \subset \Re^p$, with $\mu: \mathcal{X} \to \Re$ the mean function and $R: \mathcal{X} \times \mathcal{X} \to \Re_+$ the covariance kernel. Let $\mathcal{M}$ be a subset of the space of continuous functions mapping from $\mathcal{X}$ to $\Re$ having some constraint.  We define the projective Gaussian process $P_w$ on the constrained space $\mathcal{M}$ as
\begin{equation}
\label{eq-mother}
 P_w =\argmin_{F \in \mathcal{M}} \int_{\mathcal{X}} \{ w(t) - F(t) \}^2 dt.
\end{equation}

Let $\mathcal{M}=\mathcal{M}[0,1]^p$ denote the space of monotone functions on $[0,1]^p$. Focusing initially on the $p=1$ case, \eqref{eq-mother} has the solution
\begin{equation}
\label{projection}
P_w(x)=\inf_{v\geq x} \sup_{u\leq x} \dfrac{1}{v-u}\int_u^vw(t)dt,\; \text{for}\; x\in[0,1].
\end{equation}
The existence and uniqueness of the projection follow from Theorem 1 in Rychlik (2001).

\begin{remark}
The projection in \eqref{projection} can be well approximated using the  pooled adjacent violators algorithm (Barlow et al. (1972)).
\end{remark}

Some  properties of the projection function include:
\begin{itemize}
\item[(1)] $P_{w}(x)=w(x)$ if $w$ is a monotone function. Therefore, $P_w$ is surjective.
\item[(2)] $P_{w}(x)=c$ if $w$ is a decreasing function where $c=\int_0^1w(s)ds$ which is the slope of the line joining $(0,0)$ and $(1,\int_0^1w(s)ds)$.
\item[(3)] $P_{w}(x)$ is a continuous function given $w$ is continuous (\cite{giso}).
\end{itemize}
Hence, in projecting the Gaussian process from $C[0,1]$ to  $\mathcal{M}[0,1]$ one induces a valid measure on the set of continuous monotone functions $\mathcal{M}[0,1]$.

%We call the resulting prior on the space  $\mathcal{M}$  the Gaussian process projection prior.  For the prior on $\sigma$, assume $\pi_{\sigma}$ has a  positive continuous density on the support. For example, one can let $\sigma^{-2}$ follow a Gamma distribution $\text{Ga}(a, b)$ or $\sigma^2$  follow a log uniform distribution if $\sigma$ is compactly supported on some interval.

The following lemma on the continuity of the projection as an operator is key to showing the projective Gaussian process inherits concentration and approximation properties of the mother Gaussian process.
\begin{lemma}
\label{lemma-projection}
Let $w_1, w_2$ be continuous functions on [0,1].  Then the following holds:
\begin{align}
\label{eq-projection}
\sup_{x \in [0,1]}|P_{w_1}(x)-P_{w_2}(x)|\leq \sup_{x \in [0,1]}|w_1(x)-w_2(x)|.
\end{align}
\end{lemma}

%
%\begin{remark}
%\label{rem1}
%One can show that $P_{w}$ is continuous if $w$ is.  However,  $P_{w}$ is not  everywhere  differentiable in general.  By definition of the GCM, one can see that $T(\bar{w})(0)=\bar{w}(0)$ and $T(\bar{w})(1)=\bar{w}(1).$ The set $\left\{x:T(\bar{w})(x)<\bar{w}(x)\right\}$ is open being the union of open intervals over which  $T(\bar{w})$ are linear functions. Assuming on the contrary that $T(\bar{w})$ is not linear over some of the intervals, one can always construct a convex function above $T(\bar{w})$. Therefore, over these regions, the slopes of $T(\bar{w})$ or projection of $w$ are constants. Therefore, $P_w$ has non-differentiability at the end points of those open intervals. In the case when $\left\{T(\bar{w})<\bar{w}\right\}=\emptyset$, $w$ is monotone, then $P_w=w$. When $\left\{T(\bar{w})<\bar{w}\right\}=$(0, 1), $P_w$ equals  the constant $w(1)-w(0).$
%\end{remark}
%
 Monotone curve estimation under the projective Gaussian process is easily extended to monotone surface estimation. As an illustration, suppose that $F \in \mathcal{M}[0,1]^2$ is a monotone continuous function on $[0,1]^2$ with respect to partial orderings, so that given $s_1 \le s_2$ and $t_1 \le t_2$, $F(s_1,t_1) \le F(s_2,t_2)$. Since $\mathcal{M}[0,1]^2$ is a closed convex cone, equation (1) can be solved to obtain $P_w$ from the sample path $w \sim  \mbox{\small{GP}}(\mu,R)$ of a two-dimensional Gaussian process. Apply Algorithm 1 to obtain the solution.\\
\begin{algorithm}[!h]
\caption{} \label{al1}
%\vspace*{-12pt}
Given any fixed $t$, $w(s,t)$ is a function of $s$ and apply \eqref{projection} to obtain $\widehat{w}^{(1)}(s,t)$ by projecting $w$ along the $s$ direction.
Letting $S^{(1)} = \widehat{w}^{(1)} - w$, project $w + S^{(1)}$ onto $\mathcal{M}[0,1]$ with respect to the $t$ direction to obtain $\widetilde{w}^{(1)}(s,t)$.
 Let $T^{(1)} = \widetilde{w}^{(1)} - (w+S^{(1)})$.
 Letting $i=2,\ldots,k$, in the $i$th step  we obtain $\widehat{w}^{(i)}$ by projecting $w + T^{(i-1)}$ along the $s$ direction for any $t$ and $\widetilde{w}^{(i)}$ as the projection of $w+S^{(i)}$ along the $t$ direction for any $s$.
 The algorithm terminates when $\widehat{w}^{(i)}$ or $\widetilde{w}^{(i)}$ is monotone with respect to both $s$ and $t$ for some step $i$.
 \end{algorithm}

Theorem 1 characterizes the solution to Algorithm 1.
\begin{theorem}
\label{th-surfsol}
Let $P_{w}$ be the projection of $w$ solving \eqref{eq-mother}. Then one has
\begin{equation}
P_{w}(s,t)=\lim \widehat{w}^{(k)}(s,t)=\lim \widetilde{w}^{(k)}(s,t)\;\text{ as}\;k\rightarrow\infty,
\end{equation}
and
\begin{equation}
\sup_{s,t}|P_{w_1}(s,t)-P_{w_2}(s,t)|\leq \sup_{s,t}|w_1(s,t)-P_{w_2}(s,t)|.
\end{equation}
\end{theorem}
Theorem \ref{th-surfsol} implies that higher dimensional projections can be obtained by sequentially  projecting the adjusted $w$ along each of its directions. This approach can be trivially extended to $p>2$ dimensional problems with monotonicity constraints in one or more directions.

 %Our theorem also implies an easy implementation scheme for the  higher dimensional monotone regression problem  which is a great advantage in general multivariate regression settings.

\section{ Bayesian inference under projective Gaussian process}
\subsection{Model and notation}

We carry out Bayesian inference under the projective Gaussian process, focusing on estimation of the $p$-dimensional monotone function $F(x)$ assuming the following model
\begin{equation}
\label{eq-model}
y_i=F(x_i)+\epsilon_i, \; 1\leq i\leq n\; \;(x_1\lesssim x_2\lesssim\ldots\lesssim x_n),
\end{equation}
where $\epsilon_i \sim N(0,\sigma^2)$ and $F(x)$ is monotone increasing under partial orderings in the sense that $F(x_1) \le F(x_2)$ whenever $x_1\lesssim x_2$.  Without loss of generality, assume  $x$ lies in the compact set $[0,1]^p$.  The design of $x_1,\ldots, x_n$ can be fixed or  random where  $x_i\sim G_0$ for some distribution $G_0$.  Although we focus on Gaussian residuals for simplicity, the methods can be automatically applied in general settings.

We first define some notions of neighborhoods.
Let $\eta = (F,\sigma)$, $\eta_0 = (F_0,\sigma_0)$ denote the true value, and $\Pi$
denote the prior on $\eta$, which is expressed as $\Pi_F\Pi_{\sigma}$, where $\Pi_F$
and $\Pi_{\sigma}$ are independent priors on $F$ and $\sigma$ respectively.  As shorthand, let
$f_{xF}$ denote the conditional density $N(F(x),\sigma^2)$ with $f_{x0}$ the true
conditional density.  For random design, Hellinger distances $d_{H}(\eta,\eta_0)$ are defined as
\begin{eqnarray*}
d_H^2(\eta,\eta_0) = \int d_h^2(f_{xF},f_{x0})G_0(dx),
\end{eqnarray*}
with $d_h^2(f_{xF},f_{x0}) = \frac{1}{2} \int \Big( \sqrt{f_{xF}} - \sqrt{f_{x0}}
\Big)^2dy$. We let $U_{\epsilon}(\eta_0)$ denote  an $\epsilon$ Hellinger neighborhood around
$\eta_0$ with respect to $d_H$.  The Kullback-Leibler  divergence between $\eta$ and $\eta_0$  is
\begin{align}
\label{eq-KLdivergence}
d_{KL}(\eta,\eta_0)&=\int \int f_{x0} \log \tfrac{ f_{x0}}{f_{xF}}dyG_0(dx).
\end{align}
An $\epsilon$ Kullback-Leibler neighborhood around $\eta_0$ is denoted by $K_{\epsilon}(\eta_0)$.

\subsection{Projective Gaussian process prior}
We first use a projective Gaussian process, $F \sim \mbox{\small{pGP}}_{\mathcal{M}}( \mu, R)$, as a prior on the monotone regression function $F(x)$ for $x \in [0,1]$. We assume the mother Gaussian process $w \sim \mbox{\small{GP}}( \mu, R)$ is continuous, with $\overline{\mathcal{H}} = C[0,1]$ the closure of the reproducing kernel Hilbert space corresponding to $R$.  Our proofs assume $\mu = 0$.
The following Theorems show posterior consistency and convergence rates under our projective Gaussian process prior in the $p=1$ special case; these Theorems can be generalized to arbitrary $p$.

\begin{theorem}
\label{th-consist1}
Let $w_0$ be in the pre-image of $F_0$ so that  $F_0=P_{w_0}$. Let $\Pi_F$ be the  projective Gaussian process prior on $\mathcal{M}[0,1]$.  Assume $\Pi_{\sigma}$ has a positive continuous density including $\sigma_0$ in its support.  Then under a random design, for all $\epsilon>0$,
\begin{equation}
\label{eq-consist1}
\Pi\big\{U_{\epsilon}^C(\eta_0)|(x_1,y_1),\ldots, (x_n, y_n)\big\}\rightarrow 0\; a.s.\;  \prod_{i=1}^nP_{f_{x_i0}},
\end{equation}
where $U_{\epsilon}^C(\eta_0)$ is the complement of $U_{\epsilon}(\eta_0)$ in $\mathcal{M}[0,1]$.
\end{theorem}

A similar consistency theorem holds for the case of fixed designs.
\begin{theorem}
\label{th-consist2}
 We maintain the same conditions as in Theorem \ref{th-consist1}.  For the case of a fixed design, the posterior under the  projective Gaussian process prior is consistent, that is, for all $\epsilon>0$
\begin{equation}
\label{eq-consist2}
\Pi\big\{U_{\epsilon}(\eta_0)^C|(x_1,y_1),\ldots, (x_n, y_n)\big\}\rightarrow 0 \; a.s.\;  \prod_{i=1}^nP_{f_{x_i0}},
\end{equation}
where $U_{\epsilon}(\eta_0)$ is the average (empirical) Hellinger neighborhood.
\end{theorem}

%Posterior consistency was previously shown for alternative priors for monotone
%regression in Shively et al. (2009) and Shively et al. (2011), and it is well known that consistency is a weak property.
%Characterizing rates of contraction of the posterior around $\eta_0$ would provide a
%much stronger result.
%
% As shown in  \cite{van1}, the rates of contraction of a general zero-mean Gaussian process prior are governed by the following concentration function
%\begin{align}
%\label{eq-ratefunction}
%\phi_{w_0}(\epsilon)=\inf_{h\in\mathcal{H}:||h-w_0||<\epsilon}||h||^2_{\mathcal{H}}-\log P(||W_t||<\epsilon)
%\end{align}
%which is basically controlled by the approximation of the  true quantity $w_0$ by the elements in $\mathcal{H}$ and the small ball probability of the Gaussian process $W_t$.

  %We combine the result from Lemma \ref{lemma-projection} and the results in \cite{van1} to show the following Theorem.
For the rates theorem, we assume a fixed design and that $\sigma^2$ follows a log-uniform prior $\Pi_{\sigma^2}$ on a compact interval $[l,u]$ including $\sigma_0^2$ with $l>0$. Let $\phi_{w_0}(\epsilon_n)$ denote the Gaussian process concentration function defined in the Appendix.
\begin{theorem}
\label{th-rates1}
Let $F_0$ be the true monotone function and $w_0$ be any element in the pre-image of $F_0$.   Let $\Pi_F= \mbox{\small{pGP}}_{\mathcal{M}}( 0, R)$.  If $\phi_{w_0}(\epsilon_n)\leq n\epsilon_n$,  $\Pi_{\sigma^2}\left\{ \sigma_0^2\left(1-\epsilon_n^2/3,  1+\epsilon_n^2/3 \right) \right\}\geq e^{-C_1n\epsilon_n^2}$  and $\tfrac{u-l}{2l^2}\tfrac{1}{\epsilon_n^2}\leq e^{C_0n\epsilon_n^2}$ for some constants $C_1$ and $C_0$, the posterior distribution of $\eta$  satisfies
$$\Pi_n\{\eta: d_{H}(\eta,\eta_0)>M\epsilon_n|(x_1,y_1),\ldots,(x_n,y_n)\}\rightarrow 0  \; a.s.\;  \prod_{i=1}^nP_{f_{x_i0}}$$
for $M$ large enough where $d_H(\cdot, \cdot)$ is the empirical Hellinger distance.
\end{theorem}
%Here $\epsilon_n$ is called the posterior convergence rate of $\Pi_n(\cdot|(x_1,y_1),\ldots,(x_n,y_n)).$

Let $W_t= \mbox{\small{GP}}(0, R)$ with squared exponential covariance kernel $R(t_1,t_2)=e^{-(t_1-t_2)^2}$.  Define a scaled Gaussian process $W^A=\left(W_{At}\right)$. As an example, we consider the rate of contraction for the projection prior using $W^A$ with $A$ having a Gamma prior as in van der Vaart $\&$ van Zanten (2007, 2009).
\begin{corollary}
\label{coro1}
Let $\Pi=\Pi_F\Pi_{\sigma}$ with $\Pi_F$  the projective Gaussian process prior induced from the projection of $W^A$.   One has the following results on the convergence rate  of the posterior.
\begin{itemize}
\item [(1)] If the true monotone function $F_0\in C^{\alpha}[0,1]\bigcap \mathcal{M}[0,1]$ for some $\alpha\geq 0$, then the posterior converges at rate at least $n^{-\alpha/(2\alpha+1)}(\log n)^{(4\alpha+1)/(4\alpha+2)}$.
%\item [(2)] If $F_0\in C^{\infty}([0,1] $ such as the flat functions.
\item [(2)] If $F_0\in C^{0}[0,1]\bigcap \mathcal{M}[0,1]$ which is continuous but not differentiable, then the convergence rate is  at least $n^{-\alpha/(2\alpha+1)}(\log n)^{(4\alpha+1)/(4\alpha+2)}$($\alpha\geq1$), if there exists $w_0\in  C^{\alpha}[0,1]$ such that $P_{w_0}=F_0$.
    %\item [(3)] If $F_0$ is a flat function, i.e., $F_0=C$ for some constant $C$, then the rate of convergence is $n^{-1/2}(\log n)^2.$
\end{itemize}
\end{corollary}

%\begin{remark}
% The minimax rate for one-dimensional functions in an unconstrained H\"older class with smoothness $\alpha$ is $n^{-\alpha/(2\alpha + 1)}$. Result (1) in Corollary \ref{coro1} show that the Gaussian process projection prior achieves this rate up to a log factor.
%\end{remark}
%
%\begin{remark}
%The projection is not one to one and  $P_w$ is in general not everywhere differentiable since some of the functions lose  smoothness after projection. One consequence of  result (2) in  Corollary \ref{coro1} is that,  given $F_0$ only continuous but non-differentiable, the convergence rates of $F_0$ depend on the smoothest element $w_0$ in the pre-image of $F_0$.   Hence, if we can find an element in the pre-image set that is smoother than $F_0$ we may obtain a faster convergence rate than would have been theoretically possible (in a minimax sense) without the constraint. Rather than developing new mathematical techniques, by combining the main results of Lemma \ref{lemma-projection}, the projection properties and the rates results in Van der Vaart $\&$ van Zanten (2008, 2009),  Theorem \ref{th-rates1} and Corollary \ref{coro1}  convey the message that it is possible to achieve faster rates than the possible minimax rates for the non-constrained prior by imposing a prior on a smaller constrained space.
%\end{remark}
%
%

\begin{remark}
\label{rem-norm}
If the covariates are random from $G_0$, the rates hold with  norm $d_2(\eta,\eta_0)=\left[\int_0^1\{F(x)-F_0(x)\}^2G_0(dx)\right]^{1/2}+|\sigma-\sigma_0|$.
\end{remark}

\subsection{Inference by projecting the Gaussian process posterior}

In this section we propose an alternative approach that relies on
projecting draws from the  posterior under a Gaussian process prior onto the
space of monotone functions $\mathcal{M}[0,1]^p$. This approach is easy to implement and has excellent performance in applications we have considered.

 We first impose a  Gaussian process on $F$ and a prior $\Pi_{\sigma}$ on $\sigma$, and then project the posterior of $F$ onto $\mathcal{M}[0,1]^p$.  This induces a probability measure on $\mathcal{M}[0,1]^p$ based on which the inference is carried out. We denote by $\widetilde{\Pi}(\cdot|y_1,\ldots, y_n)$  the induced distribution on $\Gamma=\mathcal{M}[0,1]^p\times (0,\infty)$ . We first present the following Theorem which shows the existence of a prior on  $\Gamma$  whose posterior is $\widetilde{\Pi}(\cdot|y_1,\ldots, y_n)$. Hence,  our inference scheme fits in the Bayesian paradigm. Assume $\sigma$ is compactly supported.

\begin{theorem}
 Given $\widetilde{\Pi}(\cdot|y_1,\ldots, y_n)$, a probability measure on $\Gamma$  obtained by projecting the posterior of a Gaussian process onto $\mathcal{M}[0,1]^p$, there  exists a prior $\widetilde{\Pi}(dF, d\sigma)$  on $\Gamma$ whose posterior is $\widetilde{\Pi}(\cdot|y_1,\ldots, y_n)$.
\end{theorem}

Let $F_0\in \mathcal{M}[0,1]^p$ be the true monotone function.    Let $\eta=(F,\sigma)$ and $\eta_0=(F_0, \sigma_0)$. In proving the theory, we consider the random design with the covariates sampled from a distribution $G_0$ with distance $d_2(\eta,\eta_0)$ the same as in Remark \ref{rem-norm}. Theory is shown for the special case $p=1$ which can be generalized to arbitrary $p$.
%\begin{equation}
%d_2(\eta, \eta_0)=\left[\int\{F(x)-F_0(x)\}^2G_0(dx)\right]^{1/2}+\left|\sigma-\sigma_0\right|.
%\end{equation}

Since the covariates are from a distribution $G_0$, we consider the projection of the function $w(t)$ onto the monotone space by minimizing
\begin{equation}
\label{eq-randomPJ}
\int_0^1 \big\{w(t)-F(t)\big\}^2G_0(dt).
\end{equation}
The solution to \eqref{eq-randomPJ} is given by
\begin{equation}
\label{eq-randsol}
P_w(x)=\inf_{v\geq x} \sup_{u\leq x} \dfrac{1}{G_0(v)-G_0(u)}\int_u^vw(t)G_0(dt),\; \text{for}\; x\in[0,1],
\end{equation}
which is a weighted version of \eqref{projection}. In terms of implementing the projection, one can use the pooled adjacent violators algorithm with non-constant weights.
The following lemma shows continuity of the projection.
\begin{lemma}
\label{lemma-contragcm}
Let $w_1$ and $w_2$ be two functions on [0,1]. Then one has
\begin{equation}
||P_{w_1}-P_{w_2}||_{2G_0(dx)}\leq ||w_1-w_2||_{2G_0(dx)}
\end{equation}
where $||f-g||_{2G_0(dx)}=\{\int (f-g)^2G_0(dx)\}^{1/2}$.
\end{lemma}

%Now  we are ready to state the theorem on the contraction rates of $\widetilde{\Pi}(\cdot|y_1,\ldots, y_n)$. Let $\tilde\eta=(F, \sigma)$ with $F\in \mathcal{M}$ and $\tilde\eta_0=(F_0, \sigma_0)$.
%
%\begin{theorem}
%   On assuming the same conditions on $\epsilon_n$ as in Theorem \ref{th-posteriorProject},
%$$\widetilde{\Pi}_n\left(\tilde\eta:d_2(\tilde\eta, \tilde\eta_0)>M\epsilon_n|(x_1,y_1),\ldots, (x_n,y_n)\right)\rightarrow 0\; a.s. \; \prod_{i=1}^n P_{x_i0},$$
%where  $d_2(\tilde\eta, \tilde\eta_0)=\left(\int_0^1(F(x)-F_0(x))^2G_0(dx)\right)^{1/2}+\left|\sigma-\sigma_0\right|$ for all $F\in \mathcal{M}.$
%\end{theorem}

%\begin{proof}
%$\widetilde{\Pi}(\cdot|y_1,\ldots, y_n)$ is  induced from $\Pi_n(\cdot|y_1,\ldots, y_n)$  by projection.
%    By Lemma \ref{lemma-contragcm}, one can show that
%    $$d_2\left((F, \sigma), (F_0, \sigma_0)\right)\leq d_2\left((w,\sigma), (F_0,\sigma_0)\right),$$
%    where $w$ is any preimage of $F$ under the projection which includes at least the monotone function $F$.  Then our proof follows trivially.
%\end{proof}
%

\begin{theorem}
Given the scaled Gaussian process $W^A$ with $A$ from some Gamma distribution, the convergence rates of $\widetilde{\Pi}$ with respect to $d_2(\eta, \eta_0)$ are given as follows:
\begin{itemize}
\item [(1)] If the true monotone function $F_0\in C^{\alpha}[0,1]\bigcap \mathcal{M}[0,1]$, then the posterior converges at rate at least $n^{-\alpha/(2\alpha+1)}(\log n)^{(4\alpha+1)/(4\alpha+2)}$.
    \item [(2)] If $F_0$ is a flat function, so that $F_0=C$ for some constant $C$, then the rate of convergence is at least $n^{-1/2}(\log n)^2.$
\end{itemize}
\end{theorem}
Our theory of projecting the posteriors applies naturally to the higher dimensional cases.

\section{Posterior computation}
\label{simu}
  \subsection{ Monotone estimation of curves with simulated data}
 We apply the approach proposed in \S3$\cdot$3 by projecting the posterior of a Gaussian process.  Let $w\sim \mbox{\small{GP}}(0, R)$ with $R(x_1,x_2) = \beta^{-1}\exp\{ -\gamma (x_1-x_2)^2 \}$,  where $\beta \sim \mbox{Ga}( 4,1)$, $\gamma \sim \mbox{Ga}(4,1)$ and $\sigma^{-2} \sim \mbox{Ga}(4,1)$.  In
a first stage, we run a Markov chain Monte Carlo algorithm to obtain draws from the joint posterior of covariance parameters
$(\beta,\gamma,\sigma)$ and the pre-projection curve evaluated at the data points $ w_n^* = \{
w(x_1),\ldots,w(x_n) \}$.  This can proceed using any of a wide variety of algorithms developed for Gaussian process regression models; we use Vanhatalo et al. (2012, arXiv:1206.5754v1). The number of Markov chain Monte Carlo iterations is taken to be 5,000 with a burn in of 1,000.  After convergence, sample paths $ w_n^* $ are then projected to the monotone space using the pooled adjacent violators algorithm.

Data of size $n=100$ are simulated from a normal error model with standard deviation $\sigma=1$. The true mean functions given below are proposed by \cite{holmes} and \cite{neelon} and  are also used in a comparative study in  \cite{shiverly09}.
  \begin{itemize}
  \item [(a)] $F_1(x)=3$, $x\in(0,10]$ (flat function).
  \item [(b)] $F_2(x)=0.32\{x+\sin(x)\}$, $x\in(0,10]$ (sinusoidal function).
  \item [(c)] $F_3(x)=3$ if $x\in(0,8]$ and $F_3(x)=6$ if $x\in (8,10]$ (step function).
  \item [(d)] $F_4(x)=0.3x$, $x\in(0,10]$ (linear function).
  \item [(e)] $F_5(x)=0.15\exp(0.6x-3)$, $x\in(0,10]$ (exponential function).
  \item [(f)] $F_6(x)=3/\left\{1+\exp(-2x+10)\right\}$, $x\in(0,10]$  (logistic function).
  \end{itemize}

The $x$ values are taken to be equidistant in the interval $(0,10]$. The root mean squared error of the estimates is calculated in the simulation study for the Gaussian process with and without projection, with the results shown in Table 1.  The results  presented in the following table are the average root mean squared error of 50 samples of data. We compare our results with  the root mean squared error results of the regression spline provided in \cite{shiverly09}.
\begin{table}[ht]
\caption{Root mean square error for simulated data with $n=100$ and the results averaged across 50 simulation replicates in each case}
\begin{tabular}{|p{7.2em}|p{3.5em}p{4em}p{3.5em}p{4em}p{4.5em}p{3.8em}|}
\hline
    & flat & sinusoidal & step & linear & exponential & logistic \\
    \hline
    Gaussian process         & 0.151    & 0.219    & 0.271     &  0.167     &0.197     & 0.255    \\
\hline
Gaussian process projection & 0.113 & 0.211  &  0.253  &  0.163    &  0.191  &  0.224 \\
\hline
regression spline &0.097  & 0.229  &  0.285 & 0.240  & 0.213 & 0.194\\
\hline
\end{tabular}
\label{mse}
\end{table}
Figs \ref{sinu} and \ref{expo} show projection estimates and 99\% pointwise credible intervals for some of the regression functions and randomly selected simulated data sets along with the true curves. In each case the estimated curve was close to the truth and 99\% intervals mostly enclosed the true curves.

%
%\begin{figure}[ht]
%\caption{Solid lines correspond to true curves, x-marks are data points, dashed lines are posterior mean curves under the Gaussian process projection, and dot lines  are 95$\%$ pointwise intervals}
%\begin{center}
%\includegraphics[width=9cm]{F2CI1_BW}
%\label{sinu}
%\end{center}
%\end{figure}
%
%\begin{figure}[ht]
%\caption{Solid lines correspond to true curves, circles are data points, solid lines are posterior mean curves under the Gaussian process projection, dashes with dots are 95$\%$ pointwise intervals}
%\begin{center}
%\includegraphics[width=11.5cm]{expo_graypng}
%\label{expo}
%\end{center}
%\end{figure}

\begin{figure}[ht]
\caption{Dash lines correspond to true curves, circles are data points, solid lines are posterior mean curves under the Gaussian process projection, dashes with dots are 99$\%$ pointwise credible intervals.}
\includegraphics[width=10cm]{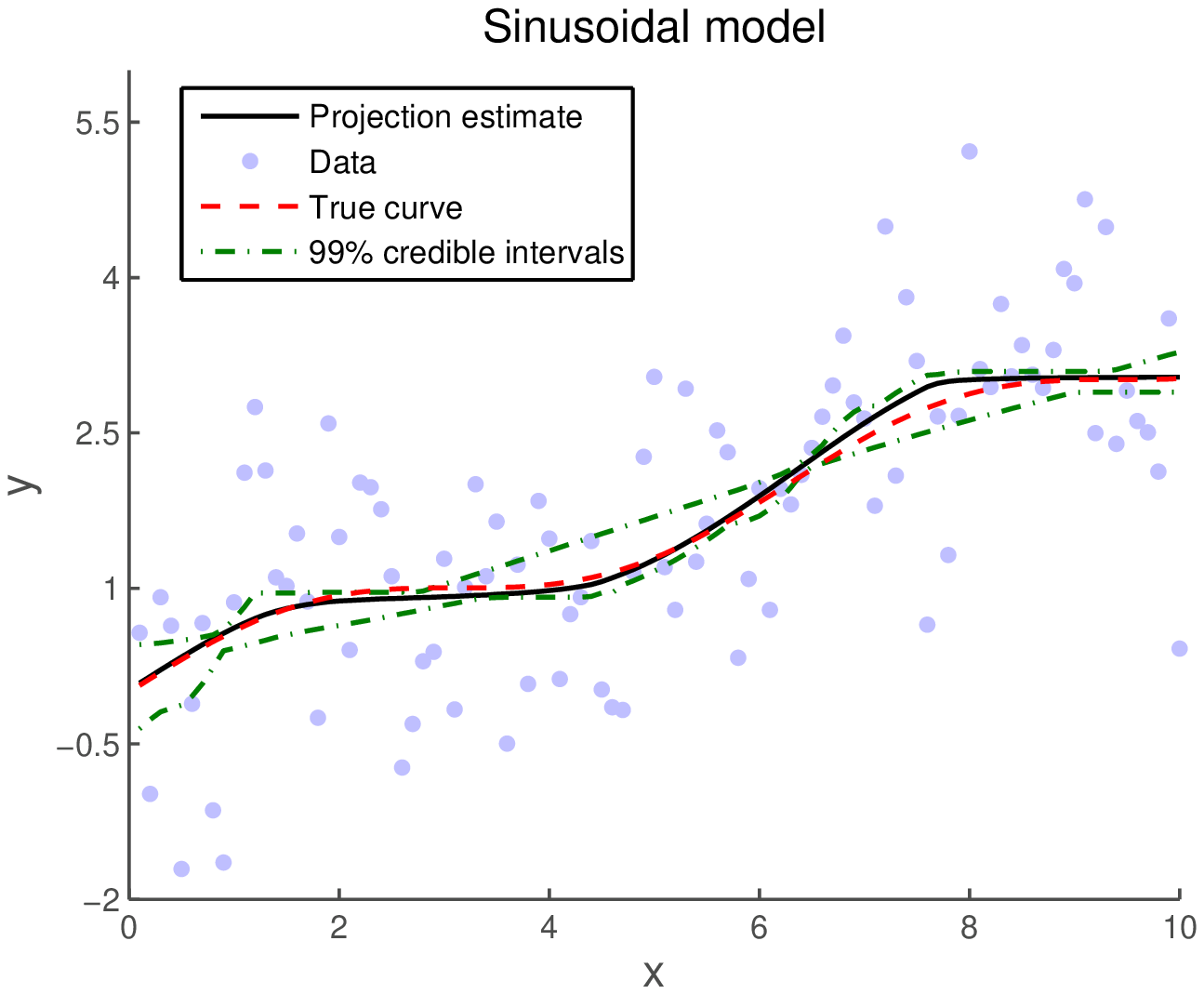}
\label{sinu}
\end{figure}

%\begin{figure}
%\caption{Dash lines correspond to true curves, circles are data points, solid lines are posterior mean curves under the Gaussian process projection, dashes with dots are 99$\%$ pointwise credible intervals.}
%% The arguments in the next line are {height}{optional width}{used only by OUP for typesetting}[filename, in directory art]
%\figurebox{20pc}{25pc}{}[color_sinu.eps]
%%\figurebox{12pc}{12pc}{F7_E1.eps}
%%\centerline{\includegraphics[width=6 cm,height=6 cm]{F7_E1.eps}}
%% note that files may not be rotated
%\label{sinu}
%\end{figure}

%
%\begin{figure}
%\caption{Dash lines correspond to true curves, circles are data points, solid lines are posterior mean curves under the Gaussian process projection, dashes with dots are 99$\%$ pointwise credible intervals.}
%% The arguments in the next line are {height}{optional width}{used only by OUP for typesetting}[filename, in directory art]
%\figurebox{20pc}{25pc}{}[color_expo.eps]
%%\figurebox{12pc}{12pc}{F7_E1.eps}
%%\centerline{\includegraphics[width=6 cm,height=6 cm]{F7_E1.eps}}
%% note that files may not be rotated
%\label{expo}
%\end{figure}

\begin{figure}[ht]
\caption{Dash lines correspond to true curves, circles are data points, solid lines are posterior mean curves under the Gaussian process projection, dashes with dots are 99$\%$ pointwise credible intervals.}
\includegraphics[width=10cm]{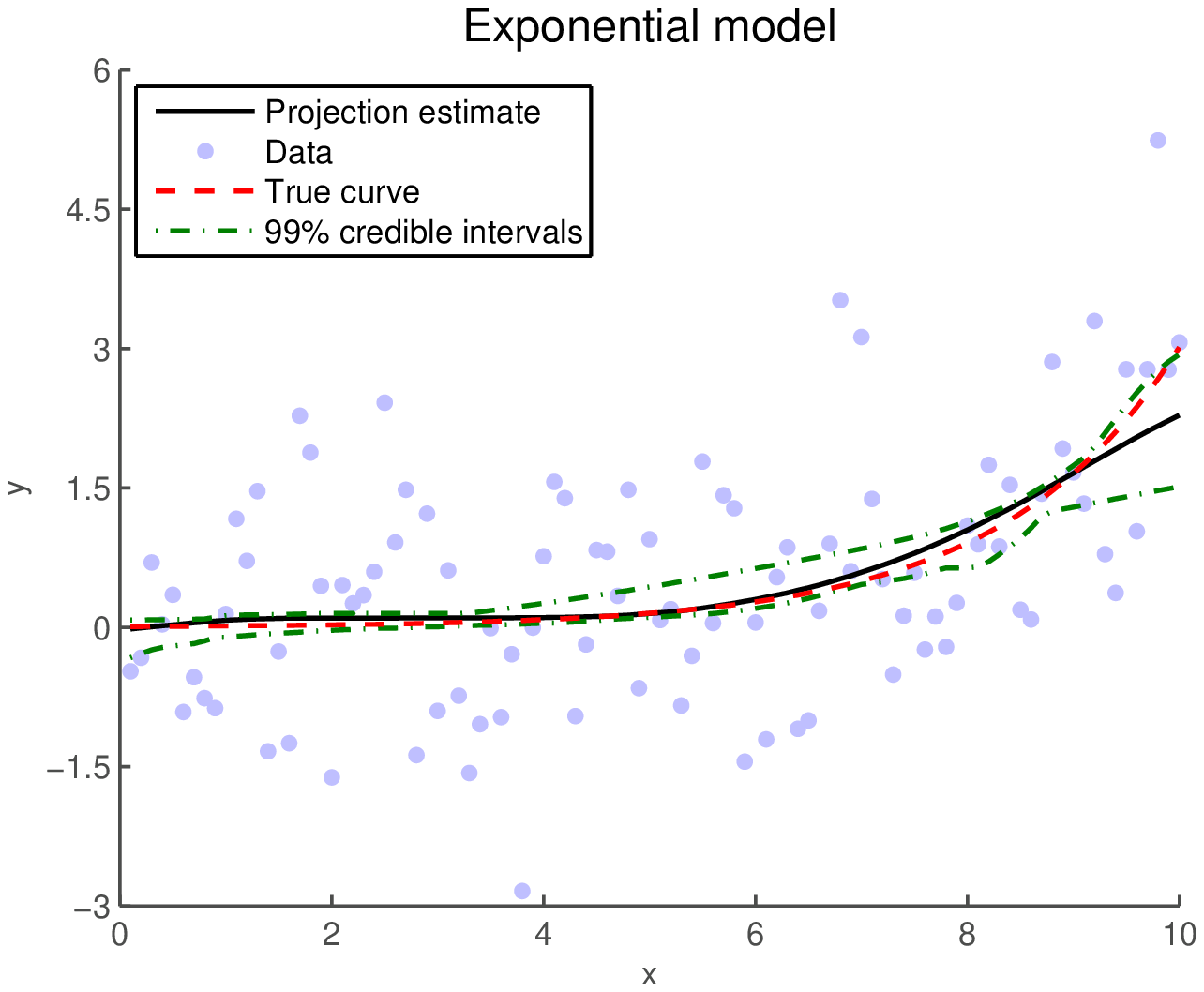}
\label{expo}
\end{figure}

  \subsection{Monotone estimation of surfaces }
%As for the projection, we use the generalized pool adjacent violators algorithm proposed by \cite{pava} which minimizes the mean square error under partial orderings.
 %A monotone surface $F(\boldsymbol x)$ is monotone with respect to each argument of $\boldsymbol x$.

In this section, we consider estimation of monotone surfaces. We choose a Gaussian process prior with covariance kernel $R(x,  x')=\beta^{-1}\exp\{ -\sum_{k=1}^2 \gamma_k(x_k-x_k')^2\}$, the posteriors of which are then projected to the monotone space.  The hyperpriors are independent with  $\beta \sim \mbox{Ga}(4,1)$, $\gamma_1 \sim \mbox{Ga}(4,1)$, $\gamma_2 \sim \mbox{Ga}(4,1)$ and $\sigma^{-2}\sim \mbox{Ga}(4,1)$.   The Markov chain Monte Carlo algorithm was run for 3,000 iterations, with the initial 500 iterations discarded. The pre-projection curve is first evaluated at the $m_1m_2$ points  $w(s_i, t_j)$ with $i=1,\ldots m_1$ and $j=1,\ldots m_2$ which are then projected to the space of monotone surfaces. We briefly describe the projection scheme in the following steps in which $w$ is only evaluated at the points $(s_i,t_j)$. This projection scheme was first introduced in \cite{roberston88} for their matrix partial ordering data.
\begin{itemize}
\item[Step 1] For any $t_j$ ($j=1,\ldots, m_2$), project $w$ along the $s$ direction by applying the pooled adjacent violators algorithm  to each vector $\{w(s_1,t_j), w(s_2,t_j),\ldots, w(s_{m_1},t_j)\}$. Denote the projection of $w$ by $\widehat{w}^{(1)}$.  Calculate the residual $S^{(1)}=\widehat{w}^{(1)}-w.$

\item[Step 2] F For any $s_i$ ($i=1,\ldots, m_1$), project $w+S^{(1)}$ along the $t$ direction using the pooled adjacent violators algorithm, calculate the residual $T^{(1)}=\widetilde{w}^{(1)}-(w+S^{(1)}).$

\item[Step 3] F  Iterate Step 1 and Step 2 by starting projecting $w+T^{(1)}$ along the $s$ direction. In the $i$th iteration ($i=1,\cdots,k$), $\widehat{w}^{(i)}$ is obtained by projecting $w+T^{(i-1)}$ along the $s$ direction and $\widetilde{w}^{(i)}$ is obtained by projecting $w+S^{(i)}$ along the $t$ direction.

\end{itemize}
For all our examples, this algorithm, which is a finite approximation to Algorithm 1, converged to a monotone solution in under 20 iterations.
By the proof of Theorem \ref{th-surfsol}, one can show that the solution obtained using the above projection scheme converges to the solution minimizing $\sum_{i=1}^{m_1}\sum_{j=1}^{m_2}\{w(s_i,t_j)-F(s_i,t_j)\}^2$  over the class of $F$ that are monotone with respect to the partial ordering on $(s_i,t_j)$.

  In the first seven examples, data of size $n=1,024$ are simulated from  a normal error model with true error $\sigma=0.5$, 0.1 and true mean regression surfaces  $F_1$ $-$ $F_7$, the first six of which are also used in \cite{Saarela}. The model fit is checked for each of our models in terms of posterior mean $\sigma$, the standard deviation of the posterior mean residuals, the correlations between the true and the posterior mean residuals and the correlation between the true and posterior mean predicted responses. We also look at the discrepancy between the true and estimated surface in terms of the mean squared error of our estimates. The results shown in Table \ref{tfit} indicate  good model fit using our projection estimates.
For the case when $\sigma=0.1$, the estimates of some models are plotted below in Figures \ref{figs1} and \ref{figs2} with the corresponding true surfaces. More plots are available in the supplementary appendix including more models and the higher noise case with $\sigma=0.5$.
  \begin{table}[ht]
\caption{  $\sigma$-true normal error; $F$-true surface; $\bar{\sigma}$-posterior mean $\sigma$; $SD(\bar{\epsilon})$-standard deviation of posterior mean residuals; cor($\epsilon,\bar{\epsilon})$-correlation between true and posterior mean residuals; cor($y,\bar{y}$)-correlation between true and posterior mean predicted responses; $mse$-mean squared error.}
\begin{tabular}{|p{2.5em}|p{2em}p{4em}p{3.5em}p{4.2em}p{4.5em}p{4em}|}
\hline
  $\sigma$  & $F$ &  $\bar{\sigma}$ & SD($\bar{\epsilon}$) & cor($\epsilon$, $\bar\epsilon$) & cor($y$, $\bar y$) & mse\\
\hline
0.5 & $F_1$ & 0.5031 & 0.5029 & 0.9962 & 0.9893 & 0.0014\\

    &  $F_2$ & 0.4993 & 0.4993 & 0.9998 & 0.9990& 0.0004\\
     &  $F_3$ & 0.4997& 0.4996 & 0.9946 & 0.9821 & 0.0016 \\
      &  $F_4$ & 0.4937 & 0.4932 & 0.9745 & 0.9463 & 0.0035\\
       &  $F_5$ & 0.5113 & 0.5110 & 0.9767 & 0.9328 & 0.0035\\
        &  $F_6$ & 0.5014 & 0.5009 & 0.9879 & 0.9758 & 0.0025\\
         &  $F_7$ &0.5034 & 0.5028 & 0.9976 & 0.9907 & 0.0011 \\
\hline
0.1 & $F_1$ &  0.0997 & 0.0996 & 0.9909  & 0.9986 & 0.0004\\

    &  $F_2$ & 0.1050 & 0.1050 & 0.9987 & 0.9997& 0.0002\\
     &  $F_3$ & 0.1008 & 0.1007 & 0.9859 & 0.9976 & 0.0005\\
      &  $F_4$ &0.1205 & 0.1204 & 0.8457 & 0.9821 & 0.0020 \\
       &  $F_5$ & 0.1145& 0.1143 & 0.8554& 0.9808 & 0.0019\\
        &  $F_6$ & 0.1045 & 0.1044 & 0.9384 & 0.9949 & 0.0010\\
         &  $F_7$ & 0.0997 & 0.0996 & 0.9977 & 0.9996 & 0.0002 \\
         \hline
\end{tabular}
\label{tfit}
\end{table}

\begin{figure}[ht]
\caption{True surface $F_1$ and its estimate.}
\includegraphics[width=8cm]{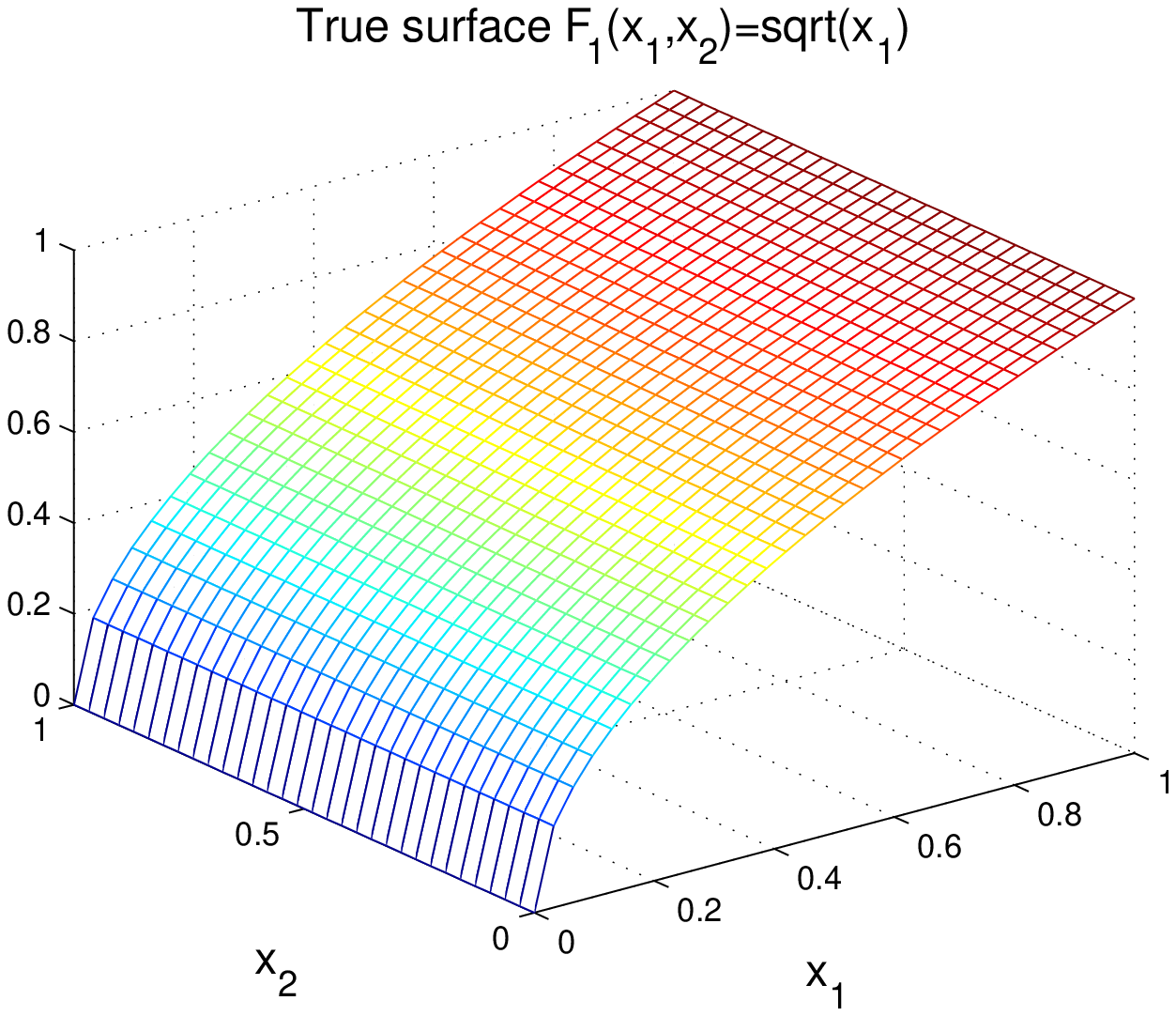}
\includegraphics[width=8cm]{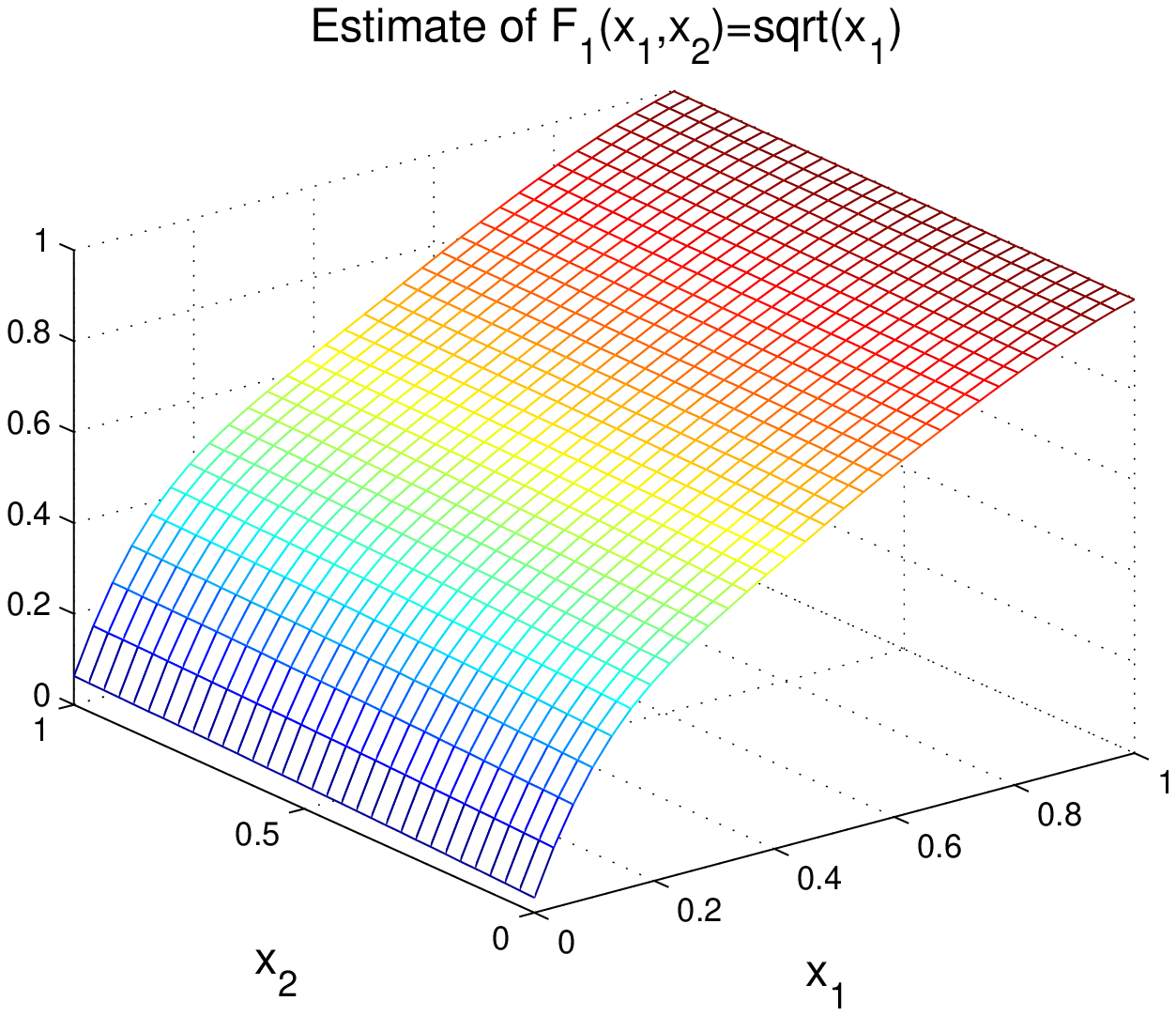}
\label{figs1}
\end{figure}

%\begin{figure}
%\caption{True surface $F_1$ and its estimate.}
%% The arguments in the next line are {height}{optional width}{used only by OUP for typesetting}[filename, in directory art]
%\parbox{\halftext}{\figurebox{15pc}{20pc}{}[F12_T1.eps]}
%\hfill
%\parbox{\halftext}{\figurebox{15pc}{20pc}{}[F12_E1.eps]}
%%\figurebox{12pc}{12pc}{F7_E1.eps}
%%\centerline{\includegraphics[width=6 cm,height=6 cm]{F7_E1.eps}}
%% note that files may not be rotated
%
%\label{figs1}
%\end{figure}

%
%\begin{figure}[htb]
%\parbox{\halftext}{% %\def\halftext{.471\textwidth}
%\figurebox{6cm}{2cm}
%\caption{The first figure on the left.}}
%\hfill
%\parbox{\halftext}{
%\figurebox{6cm}{2cm}
%\caption{The second figure on the right.}}
%\end{figure}

%\begin{figure}
%\caption{True surface $F_6$ and its estimate.}
%% The arguments in the next line are {height}{optional width}{used only by OUP for typesetting}[filename, in directory art]
%\figurebox{20pc}{25pc}{}[F62_T1.eps]
%\figurebox{20pc}{25pc}{}[F62_E1mar2.eps]
%%\figurebox{12pc}{12pc}{F7_E1.eps}
%%\centerline{\includegraphics[width=6 cm,height=6 cm]{F7_E1.eps}}
%% note that files may not be rotated
%\label{figs2}
%\end{figure}

%\begin{figure}
%\caption{True surface $F_6$ and its estimate.}
%% The arguments in the next line are {height}{optional width}{used only by OUP for typesetting}[filename, in directory art]
%\parbox{\halftext}{\figurebox{14pc}{18pc}{}[F62_T1.eps]}
%\hfill
%\parbox{\halftext}{\figurebox{14pc}{18pc}{}[F62_E1mar2.eps]}
%\label{figs2}
%\end{figure}
%

\begin{figure}[ht]
\caption{True surface $F_6$ and its estimate.}
\includegraphics[width=8cm]{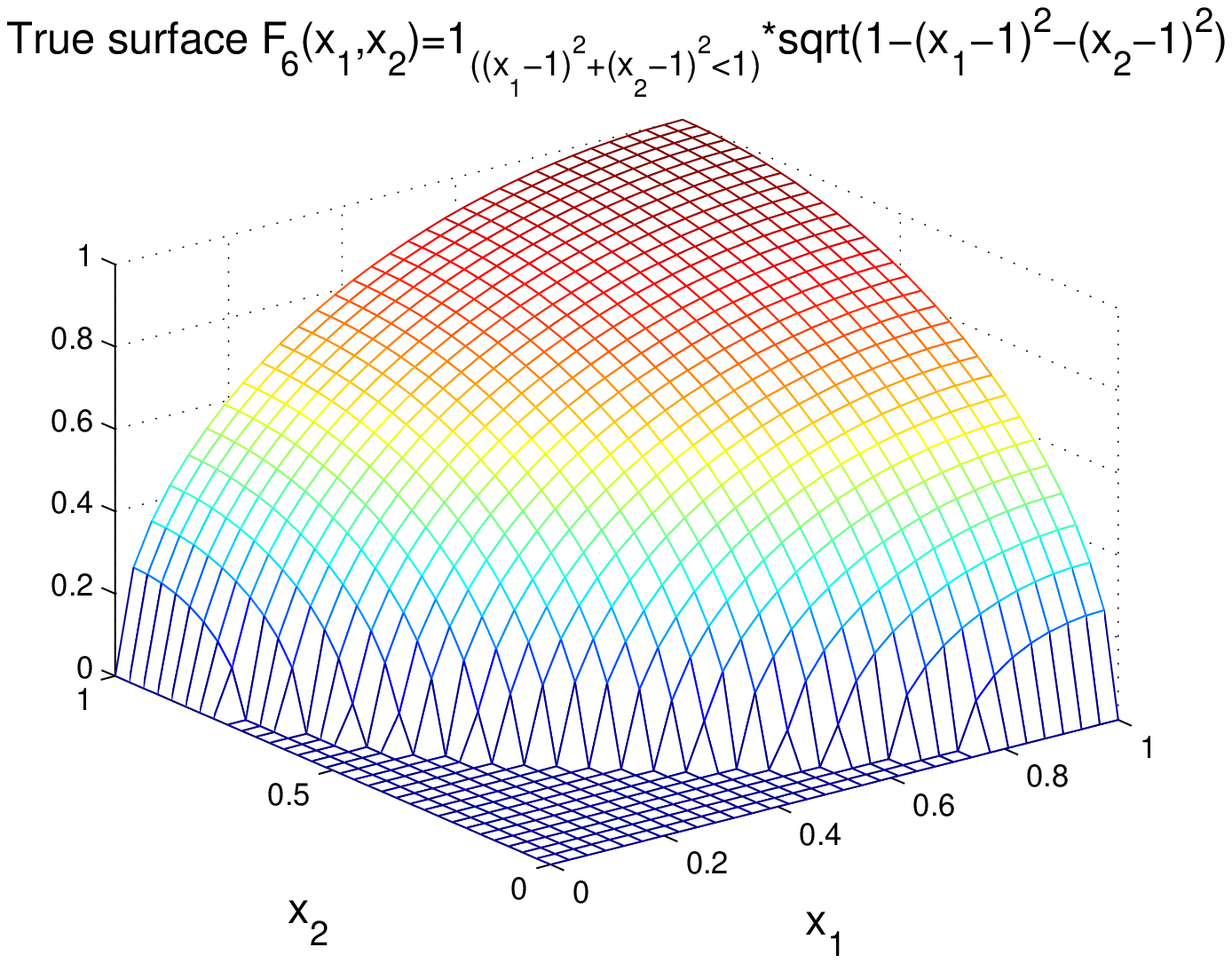}
\includegraphics[width=8cm]{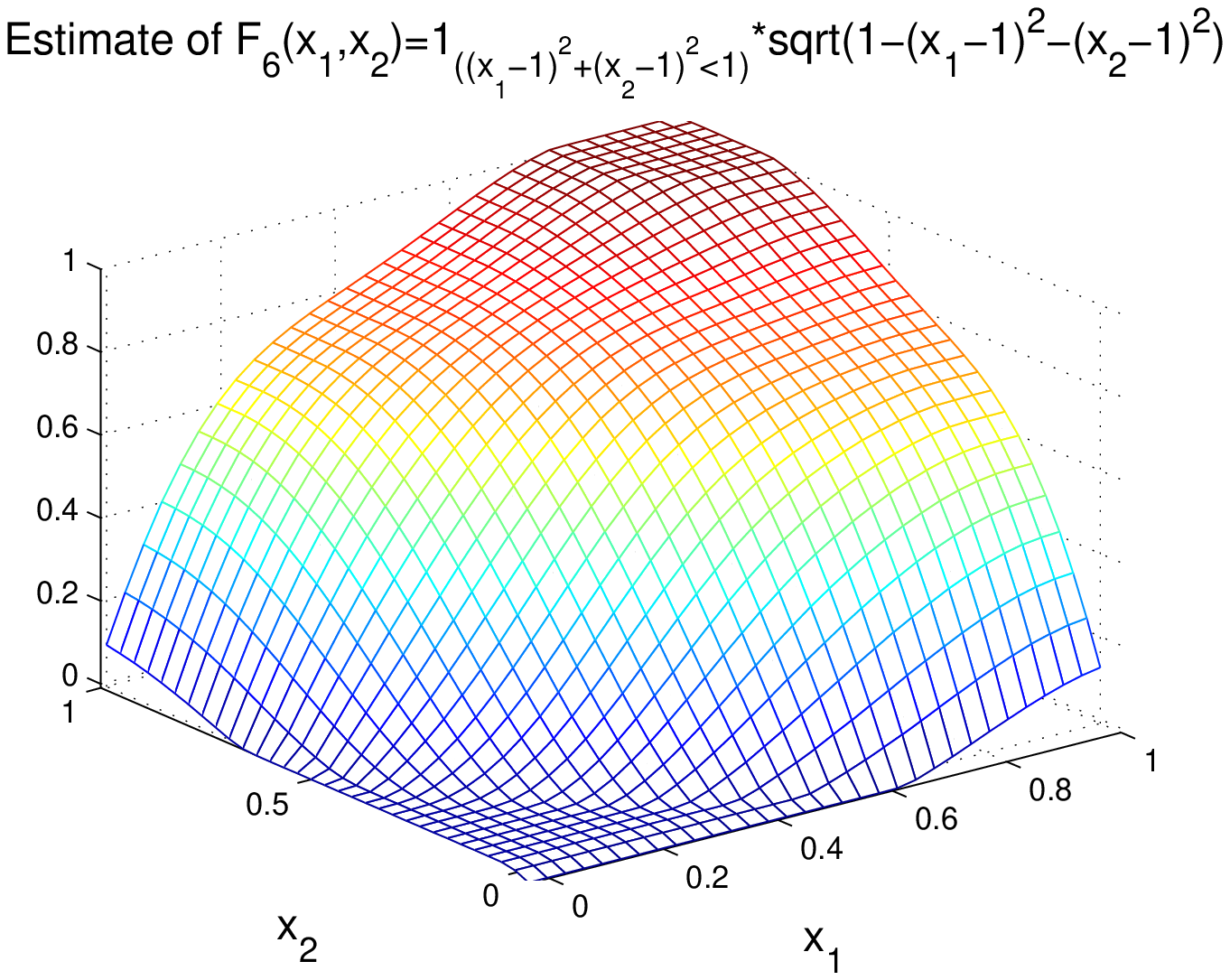}
\label{figs2}
\end{figure}

%
%\begin{figure}[ht]
%\caption{True surface $F_6$ and its estimate .}
%\includegraphics[width=6.2cm]{F62_T1}
%\includegraphics[width=6.2cm]{F62_E1mar2}
%\label{figs2}
%\end{figure}

%%
%%\begin{figure}[ht]
%%\caption{True surface F2 and its estimate .}
%%\includegraphics[width=6.2cm]{F22_T2}
%%\includegraphics[width=6.2cm]{F22_E2}
%%%\label{figs1}
%%\end{figure}
%
%
%%
%\begin{figure}[ht]
%\caption{True surface F7 and its estimate .}
%\includegraphics[width=6.2cm]{F72_T2}
%\includegraphics[width=6.2cm]{F72_E2}
%\label{figs4}
%\end{figure}
%
%  \begin{figure}[ht]
%\caption{Estimates of $F_5$ and $F_6$}
%\includegraphics[width=6.2cm]{F5}
%\includegraphics[width=6.2cm]{F6}
%\label{norm-suf}
%\end{figure}
%
%
%
%  \begin{figure}[ht]
%\caption{The estimated surface versus the true surface}
%\includegraphics[width=6.2cm]{F7}
%\includegraphics[width=6.2cm]{trueF7}
%\label{figs4}
%\end{figure}

We illustrate the application to non-Gaussian data through analyzing pneumoconiosis risk in mine workers (\cite{data}). In epidemiology and toxicology studies, it is often of interest to assess joint risk as a function of multiple exposures, with risk increasing as dose of each exposure increases. Under this assumption, the probability of pneumoconiosis is a monotone function.
%There is particular concern that there may be synergy in which the exposures adversely interact to produce a greater than additive effect.

The data were collected for coal miners who had been employed only as coal getters on the coal face and haulage workers in the underground roadways. The exposures are defined as the length of time spent at these two types of work, $t=(t_1,t_2)^T$, with records obtained on whether each miner developed pneumoconiosis. We let $\mbox{pr}(y=1| t) =
\Phi\{ F(t) \}$ which is the probability for a worker to develop pneumoconiosis under level $t$, with $F$ a real-valued bivariate monotone function and $\Phi(\cdot)$ the standard normal cumulative distribution function. We give $F$ a Gaussian process prior as described in the simulation examples, with the draws from the posterior projected to the constrained space.

  We apply our method in estimating the monotone surface of response probability. The dose-response surface is estimated by projecting  $\Phi(w)$ where $w$ is the posterior sample path of of a Gaussian process and  $\Phi(\cdot)$ is the Probit link function.  The likelihood is given by the binomial model instead of the normal model.  The estimated dose-response surface and its corresponding 95$\%$ pointwise credible intervals are plotted in Figure \ref{figmine}.

 % \begin{figure}
%  \caption{ Gaussian process projection estimate of (binary) monotone response surface and its 95$\%$ credible intervals}
%  \includegraphics[width=8cm]{Bino_dataE2}
%   \includegraphics[width=8cm]{Bino_CI2}
%  \label{figmine}
%  \end{figure}
%

%\begin{figure}
%\caption{ Gaussian process projection estimate of (binary) monotone response surface and its 95$\%$ credible intervals.}
% The arguments in the next line are {height}{optional width}{used only by OUP for typesetting}[filename, in directory art]
%\parbox{\halftext}{\figurebox{15pc}{20pc}{}[Bino_dataE1.eps]}
%\hfill
%\parbox{\halftext}{\figurebox{15pc}{20pc}{}[Bino_CI.eps]}
%\figurebox{12pc}{12pc}{F7_E1.eps}
%\centerline{\includegraphics[width=6 cm,height=6 cm]{F7_E1.eps}}
% note that files may not be rotated
%
%\label{figmine}
%\end{figure}

\begin{figure}[ht]
\caption{ Gaussian process projection estimate of (binary) monotone response surface and its 95$\%$ credible intervals.}
\includegraphics[width=8cm]{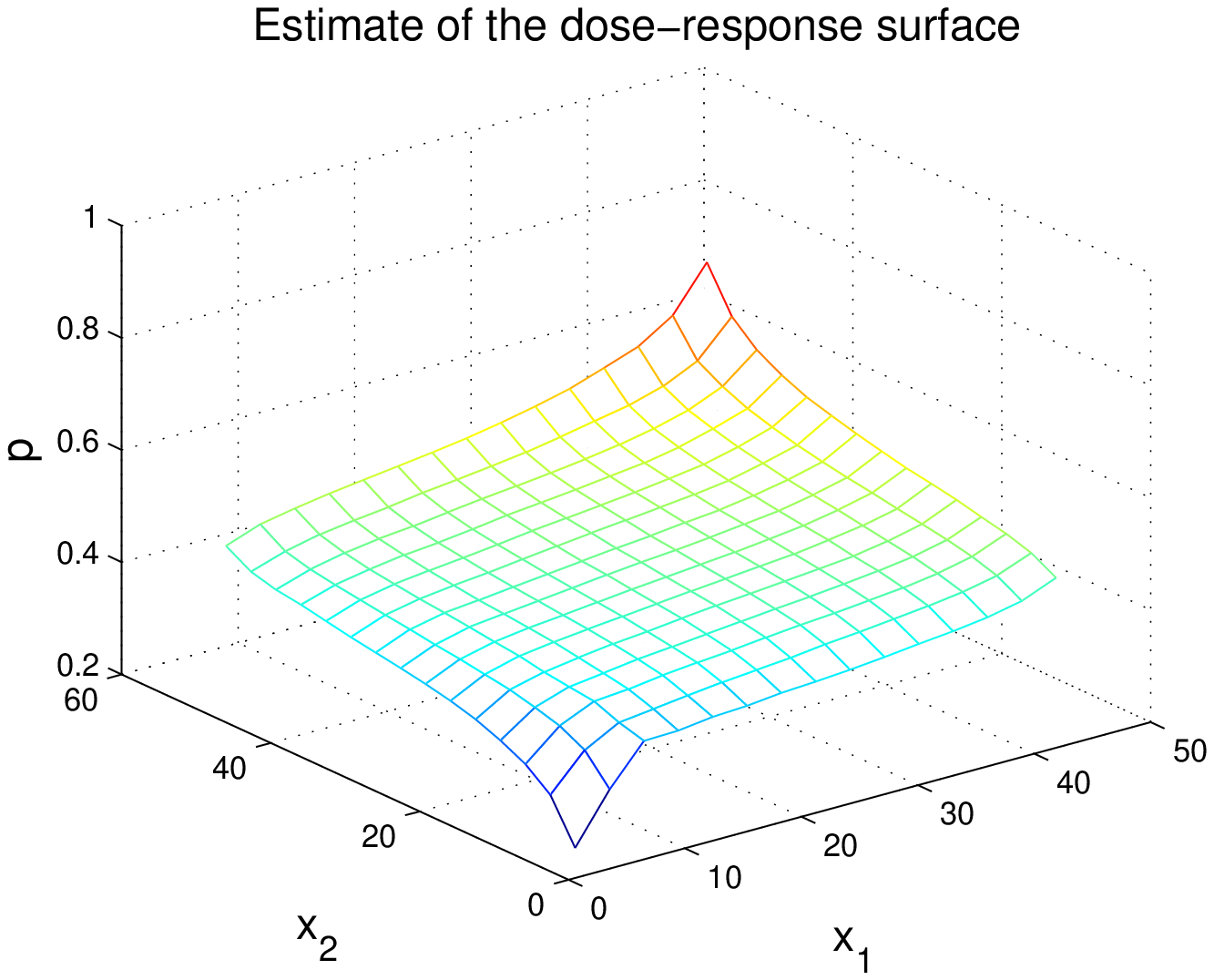}
\includegraphics[width=8cm]{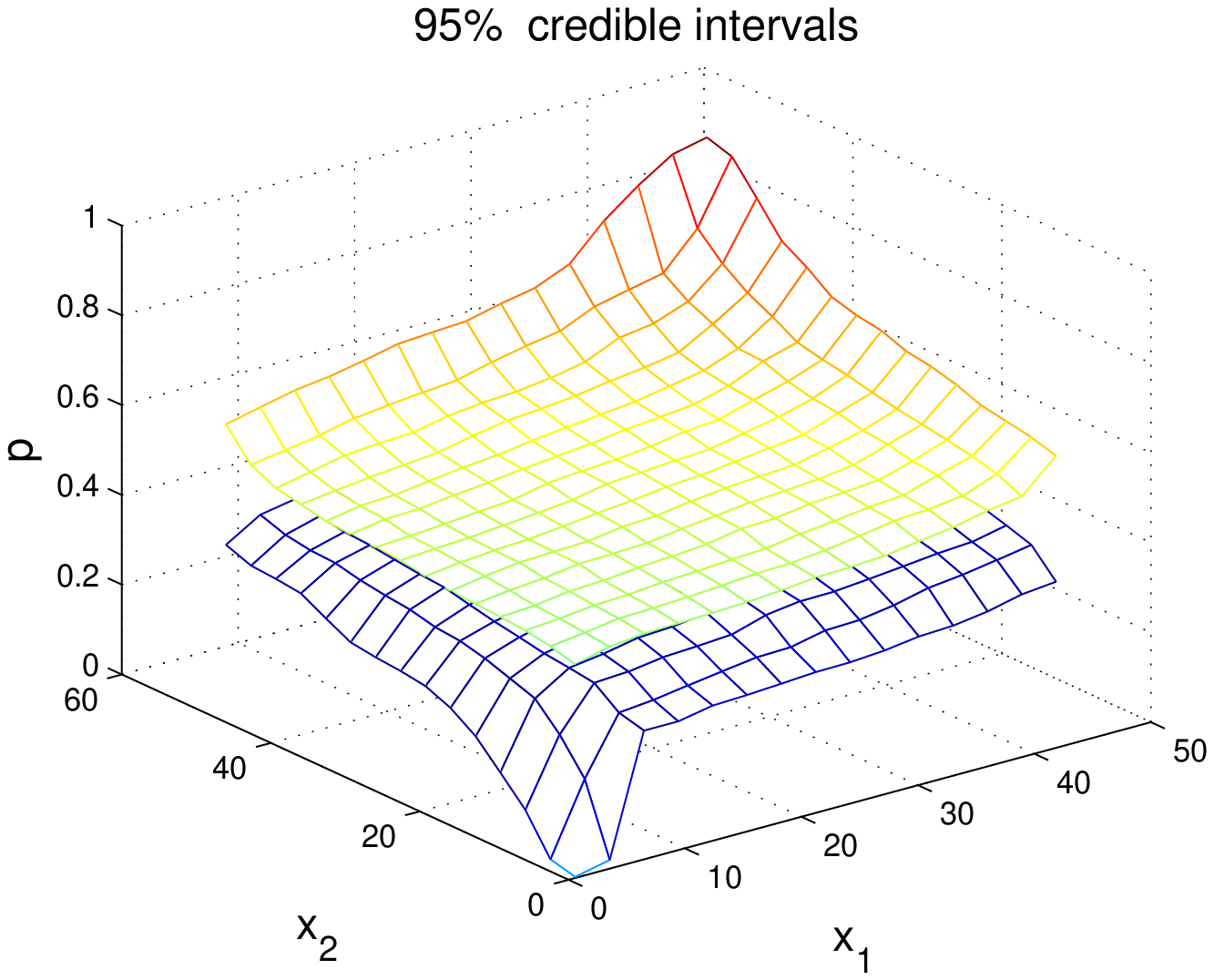}
\label{figmine}
\end{figure}

\begin{appendix}

%\appendixone
\section*{Appendix 1}
\subsection*{Lemma \ref{lemma-projection}}
\begin{proof}
Let $x$ be any real number in $[0, 1]$. We see that
\begin{align}
\label{the-difference-between-P-w1-and-P-w2-equation}
P_{w_1}(x)-P_{w_2}(x)=\inf_{v \ge x}\sup_{u \le x}\dfrac{1}{v-u}\int_u^v w_1(t)dt - \inf_{v \ge x}\sup_{u \le x}\dfrac{1}{v-u}\int_u^v w_2(t)dt.
\end{align}
For each $\epsilon > 0$, there exists an element $v_0 \ge x$ such that
\begin{align*}
\sup_{u \le x}\dfrac{1}{v_0 - u}\int_u^{v_0}w_2(t)dt < \inf_{v \ge x}\sup_{u\le x}\dfrac{1}{v - u}\int_u^v w_2(t)dt + \epsilon,
\end{align*}
and hence
\begin{align*}
-\inf_{v \ge x}\sup_{u \le x}\dfrac{1}{v-u}\int_u^v w_2(t)dt < -\sup_{u \le x}\dfrac{1}{v_0 - u}\int_u^{v_0}w_2(t)dt + \epsilon.
\end{align*}
It follows from $(\ref{the-difference-between-P-w1-and-P-w2-equation})$ that
\begin{align}
\label{the-first-inequality-for-P-w1-and-P-w2-inequality}
P_{w_1}(x)-P_{w_2}(x) < \inf_{v \ge x}\sup_{u \le x}\dfrac{1}{v-u}\int_u^v w_1(t)dt -\sup_{u \le x}\dfrac{1}{v_0 - u}\int_u^{v_0}w_2(t)dt + \epsilon.
\end{align}
Note that
\begin{align*}
\inf_{v \ge x}\sup_{u \le x}\dfrac{1}{v-u}\int_u^v w_1(t)dt \le \sup_{u \le x}\dfrac{1}{v-u}\int_u^v w_1(t)dt
\end{align*}
for all $v \ge x$. In particular, this implies that the inequality above holds for $v = v_0$. Hence we see from $(\ref{the-first-inequality-for-P-w1-and-P-w2-inequality})$ that
\begin{align}
\label{the-second-inequality-for-P-w1-and-P-w2-inequality}
P_{w_1}(x)-P_{w_2}(x) < \sup_{u \le x}\dfrac{1}{v_0 - u}\int_{u}^{v_0} w_1(t)dt -\sup_{u \le x}\dfrac{1}{v_0 - u}\int_u^{v_0}w_2(t)dt + \epsilon.
\end{align}
For each $\delta > 0$, there exists an element $u_0 \le x$ such that
\begin{align*}
\sup_{u \le x}\dfrac{1}{v_0 - u}\int_{u}^{v_0} w_1(t)dt < \dfrac{1}{v_0 - u_0}\int_{u_0}^{v_0} w_1(t)dt + \delta.
\end{align*}
Since $u_0 \le x$, we have
\begin{align*}
\sup_{u \le x}\dfrac{1}{v_0 - u}\int_u^{v_0}w_2(t)dt \ge \dfrac{1}{v_0 - u_0}\int_{u_0}^{v_0}w_2(t)dt.
\end{align*}
Thus it follows from $(\ref{the-second-inequality-for-P-w1-and-P-w2-inequality})$ that
\begin{align*}
P_{w_1}(x)-P_{w_2}(x) &< \dfrac{1}{v_0 - u_0}\int_{u_0}^{v_0} w_1(t)dt -\dfrac{1}{v_0 - u_0}\int_{u_0}^{v_0}w_2(t)dt + \epsilon + \delta \\
&< \dfrac{1}{v_0 - u_0}\int_{u_0}^{v_0}|w_1(t) - w_2(t)| dt + \epsilon + \delta \\
&< \dfrac{1}{v_0 - u_0}\int_{u_0}^{v_0}\sup_{t\in [0,1]}|w_1(t) - w_2(t)| dt + \epsilon + \delta \\
&< \sup_{t \in [0,1]}|w_1(t) - w_2(t)| + \tau,
\end{align*}
where $\tau = \epsilon + \delta$, and  $\epsilon$ and $\delta$ are arbitrarily positive numbers. Therefore,  for every $\tau > 0$ and $x \in [0,1]$, we have
\begin{align*}
P_{w_1}(x)-P_{w_2}(x) < \sup_{t \in [0,1]}|w_1(t) - w_2(t)| + \tau.
\end{align*}
Thus we see that
\begin{align*}
|P_{w_1}(x)-P_{w_2}(x)| < \sup_{t \in [0,1]}|w_1(t) - w_2(t)| + \tau
\end{align*}
for every $\tau > 0$ and $x \in T$, and hence, for every $\tau > 0$, we have that
\begin{align*}
\sup_{x \in [0,1]}|P_{w_1}(x)-P_{w_2}(x)| < \sup_{t \in [0,1]}|w_1(t) - w_2(t)| + \tau.
\end{align*}
Upon letting $\tau \rightarrow 0$ in the above inequality, the lemma follows.
\end{proof}

 We first prove a lemma which is used in proving Theorem \ref{th-surfsol}.
\begin{lemma}
\label{lem-apen}
Let $C_s$ be the cone of continuous functions $f(s,t)$ which are monotone with respect to $s$ for any $t$  and $C_t$ be the cone of continuous functions  which are monotone with respect to $t$ for any $s$.
Define their dual cones $C_s^*$ and $C_t^*$ as
\begin{equation*}
C_s^*=\left\{g(s,t)\in C[0,1]^2: \int f(s,t)g(s,t)ds\leq 0,\;\text{for all}\; t \;\text{and}\; f\in C_s  \right\},
\end{equation*}
and
\begin{equation*}
C_t^*=\left\{g(s,t)\in C[0,1]^2: \int f(s,t)g(s,t)dt\leq 0,\;\text{for all}\; s \;\text{and}\; f\in C_t  \right\}.
\end{equation*}
 Denote $P(w|C_s)$ as the projection of $w$ over $C_s$ by minimizing $\int (w-f)^2ds$ for all $f\in C_s$ and any fixed $t$. Denote  $P(w|C_t)$ as the projection of $w$ over $C_t$ by minimizing $\int (w-f)^2dt$ for all $f\in C_t$ and any fixed $s$. Then
\begin{equation}
P(w|C_s^*)=w-P(w|C_s)\;\text{and}\; P(w|C_t^*)=w-P(w|C_t).
\end{equation}
Furthermore, $P(w|C_s^*)$ turns out to be the solution to the projection by minimizing $\int (w-f)^2dsdt$ over all $f\in C_s$ and $P(w|C_t^*)$ is the solution to the projection by minimizing
$\int (w-f)^2dsdt$ over all $f\in C_t.$
\end{lemma}

\begin{proof}[Proof of Lemma \ref{lem-apen}]
First note that $P(w|C_s)$ is obtained by minimizing $\int (w-f)^2ds$ for all $f\in C_s$  and any fixed $t$.
Then according to Theorem 1 of \cite{rychlik}, one has  $\int \{w-P(w|C_s)\}fds\leq 0$ and $\int P(w|C_s)\{w-P(w|C_s)\}ds=0$ by the properties of $P(w|C_s)$. The first property implies  $w-P(w|C_s)\in C_s^*$.
 For any $h\in C_s^*$, one has $\int \{w-(w-P(w|C_s))\}hds=\int P(w|C_s)hds\leq 0$.
 One can then deduce that $P(w|C_s^*)=w-P(w|C_s).$
With a similar argument, one can show that $P(w|C_t^*)=w-P(w|C_t)$.

Since $\int \{w-P(w|C_s)\}^2ds\leq \int (w-f)^2ds$ for any fixed $t$, then one has $\int \{w-P(w|C_s)\}^2dsdt\leq \int (w-f)^2dsdt$. Therefore,  $P(w|C_s^*)$ minimizes $\int (w-f)^2dsdt$ for all $f\in C_s$ and  $P(w|C_t^*)$ minimizes $\int (w-f)^2dsdt$ for all $f\in C_t$ by the same argument.

\end{proof}

\begin{proof}[Proof of Theorem \ref{th-surfsol}]

%Let $C_s$ be the cone of functions $f(s,t)$ which are monotone with respect to $x$ for any $y$  and $C_t$ be the cone of functions  which are monotone with respect to $y$ for any $x$. $C_s^*$ and $C_t^*$ are their dual cones.
%Define their dual cones $C_s^*$ and $C_t^*$ as
%\begin{equation*}
%C_s^*=\left\{g(x,y): \int f(s,t)g(x,y)dx\leq 0,\;\text{for all}\; y \;\text{and}\; f\in C_s  \right\},
%\end{equation*}
%and
%\begin{equation*}
%C_t^*=\left\{g(x,y): \int f(s,t)g(x,y)dy\leq 0,\;\text{for all}\; x \;\text{and}\; f\in C_t  \right\}.
%\end{equation*}

Define the norm
$||f||=\langle f,f\rangle^{1/2}=\left[\int \left\{f^2(s,t)\right\}dsdt\right]^{1/2}$
with $\langle\cdot,\cdot\rangle$ denoting the inner product.

One has $-S^{(k+1)}=(w+T^{(k)})-\widehat{w}^{(k)}=(w+T^{(k)})-P(w+T^{(k)}|C_s)=P(w+T^{(k)}|C_s^*)$ where the last equality follows from  Lemma A1. Here $P(w+T^{(k)}|C_s)$ denotes the projection of $w+T^{(k)}$ onto $C_s$ and $P(w+T^{(k)}|C_s^*)$ is the projection onto $C_s^*$.
Therefore  $-S^{(k+1)}$ minimizes $||(w+T^{(k)})-f||$ for all $f\in C_s^*$ and $-T^{(k)}$ minimizes $||(w+S^{(k)})-f||$ for all $f\in C_t^*$. Then one  concludes that
$$||\widehat{w}^{(k)}||=||w+S^{(k)}-(-T^{(k-1)})||\geq ||w+S^{(k)}-(-T^{(k)})||\geq ||w+T^{(k)}-(-S^{(k+1)})||$$
for all $k$. Therefore, one has $||\widehat{w}^{(k)}||\geq||\widetilde{w}^{(k)}||\geq ||\widehat{w}^{(k+1)}||$.
Now we wish to show that $\{S^{(k)}\}$ and $\{T^{(k)}\}$ are bounded and that $||S^{(k+1)}-S^{(k)}||\rightarrow 0$ and $||T^{(k+1)}-T^{(k)}||\rightarrow 0$ as $k\rightarrow \infty$.

Assume that $\{S^{(k)}\}$ or $\{T^{(k)}\}$ is not bounded. Take an arbitrary large number $M > 0$. Then there exists an integer $N$ such that $|S^{(N)}|\geq M$ or $|T^{(N)}|\geq M$. One then partitions $[0,1]\times [0,1]$ into $m_1m_2$ squares of equal areas. The vertices of the squares are of the form $(s_i, t_j)$, where $s_i = \frac{i-1}{m_1}$ and $t_j = \frac{j-1}{m_2}$ with $i=1,\ldots,m_1+1$ and $j=1,\ldots, m_2+1$. Let $(s_{i_0},t_{j_0})$ be the point such that $|S^{(N)}|\geq M$ or $|T^{(N)}|\geq M$ over the square between $(s_{i_0},t_{j_0})$ and $(s_{i_0+1},t_{j_0+1})$ for the first time with respect to the partial ordering. Without loss of generality, assume $|S^{(N)}|\geq M$. For $t=t_{j_0}$, let $f$ be the monotone function such that $f=-1$ for $s\leq s_{i_0}$, $f=0$ for $s\geq s_{i_0+1}$ and $f$ is linearly interpolated between $(s_{i_0},t_{j_0})$ and $(s_{i_0+1},t_{j_0+1})$ . Note that for any $\epsilon>0$, one can also partition $[0,1]\times [0,1]$ finely enough such that $|\sum_{i=1}^{m_1}S(s_i, t_{j_0})f(s_{i},t_{j_0})-\int S(s,t_{j_0})f(s,t_{j_0})ds|<\epsilon$ and $|\sum_{j=1}^{m_2}S(s_{i_0}, t_j)(s_{i_0},t_j)-\int S(s_{i_0},t)f(s_{i_0},t)dt|<\epsilon$. By the properties of the dual cones and our construction of the function $f$, one can partition $[0,1]\times [0,1]$ finely enough such that $\sum_{i=1}^{i_0}S(s_i, t_{j_0})\leq 0$, which implies that $S^{(N)}\leq -M$ in the square between $(s_{i_0},t_{j_0})$ and $(s_{i_0+1},t_{j_0+1})$ up to an arbitrary small number $\epsilon$. Since the norm of $\widetilde{w}^{(N)}=w+S^{(N)}+T^{(N)}$ is bounded, it follows from the continuity of the estimates that $T^{(N)}\geq M$ up to an arbitrary small number $\epsilon$ over the square of $(s_{i_0},t_{j_0})$. On the other hand, we know that $\sum_{j=1}^{j_0}T(s_{i_0}, t_{j})\leq 0$, which contradicts the fact that $T^{(N)}\geq M$ and $(s_{i_0}, t_{j_0})$ is the point such that $|S^{(N)}|\geq M$ or $|T^{(N)}|\geq M$ over the square between $(s_{i_0},t_{j_0})$ and $(s_{i_0+1},t_{j_0+1})$ for the first time.

Since $\{S^{(k)}\}$ and $\{T^{(k)}\}$ are  bounded, then there exists convergent subsequences indexed by $n_i$ such that
$\{S^{(n_i)}\}\rightarrow S$ and $\{T^{(n_i)}\}\rightarrow T$. Then one has for $n_i\rightarrow\infty$,
\begin{equation}
\widetilde{w}^{(n_i)}=w+S^{(n_i)}+T^{(n_i)}\rightarrow w+S+T;\widehat{w}^{(n_i+1)}=w+S^{(n_i+1)}+T^{(n_i)}\rightarrow w+S+T.
\end{equation}
Denote the limit as $w_L=w+S+T.$ One claims that $w_L$ is the  projection of $w$ which is the solution to \eqref{eq-mother} under the partial ordering constraint.
First note that $w_L\in C_s$ since $\widehat{w}^{(k)}\in C_s$ and $w_L\in C_t$ since $\widetilde{w}^{(k)}\in C_t$. This implies that $w_L\in C_s\cap C_t$ which says that $w_L$ is monotone with respect to the partial ordering on $(s,t)$. Now,
\begin{align*}
\langle w-w_L,w_L\rangle&=\langle w-w_L+S,w_L\rangle-\langle S,w_L\rangle\\
&=\lim_{n_i\rightarrow\infty}\langle w+S^{(n_i)}-\widetilde{w}^{(n_i)},\widetilde{w}^{(n_i)}\rangle+
\lim_{n_i\rightarrow\infty}\langle w+T^{(n_i)}-\widehat{w}^{(n_i+1)},  \widehat{w}^{(n_i+1)}\rangle\\
&=0+0=0.
\end{align*}
Let $h$ be any element in $C_s\cap C_t$, one looks at
\begin{align*}
\langle w-w_L,h\rangle&=\langle w-w_L+S,h\rangle-\langle S,h\rangle\\
&=\lim_{n_i\rightarrow\infty}\langle w+S^{(n_i)}-\widetilde{w}^{(n_i)},h\rangle+
\lim_{n_i\rightarrow\infty}\langle w+T^{(n_i)}-\widehat{w}^{(n_i+1)},  h\rangle\\
&\leq 0+0=0.
\end{align*}
Then by Theorem 1 in \cite{rychlik}, $w_L$ is indeed the projection of $w$.
Now we will show that $||S^{(i+1)}-S^{(i)}||^2\rightarrow 0$ and $||T^{(i+1)}-T^{(i)}||^2\rightarrow 0$ with which we can conclude that $S^{(k)}\rightarrow S$, $T^{(k)}\rightarrow T$ and both  $\widetilde{w}^{(k)}$ and $\widehat w^{(k)}$ converge to $w_L$.
First by the projection property, one can show that
\begin{align*}
||T^{(i)}-T^{(i-1)}||^2&=||w+T^{(i)}-(w+T^{(i-1)}||^2\geq||S^{(i+1)}-S^{(i)}||^2\\
&=||w+S^{(i+1)}-(w+S^{(i)})||^2\geq ||T^{(i+1)}-T^{(i)}||^2.
\end{align*}
Therefore, $S^{(i+1)}-S^{(i)}$ and $T^{(i+1)}-T^{(i)}$ converge to the same limit.
Assume on the contrary that $||S^{(i+1)}-S^{(i)}||^2$ does not converge to zero. Then over some Lebesgue measure non-zero set $U$, there exists $\epsilon>0$ such that for $(s,t)\in U$
\begin{equation}
\label{eq-con1}
|S^{(i+1)}-S^{(i)}|>\epsilon\; \text{for infinitely many}\; i.
\end{equation}
Now look at
\begin{align*}
||S^{(i+1)}-S^{(i)}||^2-||T^{(i+1)}-T^{(i)}||^2&=||w+T^{(i)}+S^{(i)}-(w+T^{(i+1)}+S^{(i+1)})||^2\\
&+2\langle w+T^{(i)}+S^{(i)}-(w+T^{(i+1)}+S^{(i+1)}  ), S^{(i+1)}-S^{(i)}\rangle.
\end{align*}
Note that $||S^{(i+1)}-S^{(i)}||^2-||T^{(i+1)}-T^{(i)}||^2\rightarrow 0$ and the last term on the right hand side of the above equation is non-negative. Therefore, one can conclude that
\begin{equation}
\label{eq-ss1}
(T^{(i+1)}-T^{(i)})-(S^{(i+1)}-S^{(i)})\rightarrow 0.
\end{equation}
By a similar argument, one has
\begin{equation}
\label{eq-ss2}
(T^{(i+2)}-T^{(i+1)})-(S^{(i+1)}-S^{(i)})\rightarrow 0.
\end{equation}
Subtracting \eqref{eq-ss1} from \eqref{eq-ss2}, one has
$$(T^{(i+2)}-T^{(i+1)})-(T^{(i+1)}-T^{(i)})\rightarrow 0.$$
This implies that there exist $i,j$ large enough with $|i-j|$ finite such that
$T^{(i+1)}-T^{(i)}$ can be made arbitrarily close to $T^{(j+1)}-T^{(j)}$. However this contradicts  \eqref{eq-con1}
and the fact that $\{T^{(k)}\}$ is bounded such that there exists constant $C$ such that $|T^{(i)}-T^{(j)}|<C$ for all $i$, $j$.
By the same argument, one can show that $||T^{(i+1)}-T^{(i)}||^2\rightarrow 0$. Therefore,   $S^{(k)}\rightarrow S$, $T^{(k)}\rightarrow T$ which implies $\widetilde{w}^{(k)}\rightarrow w_L$ and $\widehat w^{(k)}\rightarrow w_L$.

The inequality in the Theorem can be  shown combining Lemma \ref{lemma-projection} and the properties of the projection.
\end{proof}

\begin{proof}[Proof of Theorem 5]

We wish to find a probability $\widetilde{\Pi}$, say, on the space $\Gamma$  (thought of as a `prior' for $(F,\sigma)$, but which may depend on the data $y_i$) such that the projection of the posterior  $\Pi(\cdot|y_1,\ldots,y_n)$  of a Gaussian process is the posterior on $\Gamma$  with prior $\widetilde{\Pi}$.  Let $\widetilde{\Pi}(dF,d\sigma|y_1,\ldots,y_n)$ denote this probability.  Since the (conditional) density of the observations $y_1,\ldots,y_n$, given $(F,\sigma)$, is the joint Normal density as before, say $f(y_1,\ldots,y_n|F,\sigma)$,  one needs to have $\widetilde{\Pi}(dF,d\sigma|y_1,\ldots,y_n)$ satisfy
  \begin{align}
  \label{eq-empi_prior}
\widetilde{\Pi}(dF, d\sigma|y_1,\ldots,y_n) \int_{  \Gamma}  f(y_1,\ldots,y_n|F,\sigma) \widetilde{\Pi}(dF, d\sigma)=  f(y_1,\ldots,y_n |F,\sigma) \widetilde{\Pi}(dF, d\sigma).
\end{align}

Let \begin{align*}
g(F,\sigma| y_1,\ldots,y_n )=\dfrac{1}{f(y_1,\ldots,y_n|F,\sigma)}\left\{\int_{  \Gamma} f\left(y_1,\ldots,y_n|F,\sigma\right)^{-1}  \widetilde{\Pi}(dF,d\sigma|y_1,\ldots,y_n)\right\}^{-1},
\end{align*}
which is well-defined since $\int_{  \Gamma} f\left(y_1,\ldots,y_n|F,\sigma\right)^{-1}  \widetilde{\Pi}(dF,d\sigma|y_1,\ldots,y_n)<\infty$.
First note that $g$ is a density on $\Gamma$ with respect to the measure  $\widetilde{\Pi}(dF,d\sigma|y_1,\ldots, y_n)$ since one can easily check that
$\int_{\Gamma} g(F,\sigma| y_1,\ldots,y_n )\widetilde{\Pi}(dF,d\sigma|y_1,\ldots, y_n)=1$.
Define $$\widetilde{\Pi}(dF, d\sigma)=g(F,\sigma| y_1,\ldots,y_n )\widetilde{\Pi}(dF,d\sigma|y_1,\ldots, y_n).$$  We will show that $\widetilde{\Pi}(dF, d\sigma)$ satisfies equation \eqref{eq-empi_prior} above which is equivalent to showing
$$ \int_{ \Gamma}  f(y_1,\ldots,y_n|F,\sigma) \widetilde{\Pi}(dF, d\sigma)=f(y_1,\ldots,y_n |F,\sigma)g(F,\sigma| y_1,\ldots,y_n ).$$
One has
\begin{align*}
 &\int_{  \Gamma}  f(y_1,\ldots,y_n|F,\sigma) \widetilde{\Pi}(dF, d\sigma)\\
 &= \int_{  \Gamma}  f(y_1,\ldots,y_n|F,\sigma)g(F,\sigma| y_1,\ldots,y_n )\widetilde{\Pi}(dF,d\sigma|y_1,\ldots, y_n)\\
 &=\int_{ \Gamma}  \left\{\int_{  \Gamma} f\left(y_1,\ldots,y_n|F,\sigma\right)^{-1}  \widetilde{\Pi}(dF,d\sigma|y_1,\ldots,y_n)\right\}^{-1}\widetilde{\Pi}(dF,d\sigma|y_1,\ldots, y_n)\\
 &= \left\{\int_{ \Gamma} f\left(y_1,\ldots,y_n|F,\sigma\right)^{-1}  \widetilde{\Pi}(dF,d\sigma|y_1,\ldots,y_n)\right\}^{-1}=f(y_1,\ldots,y_n |F,\sigma)g(F,\sigma| y_1,\ldots,y_n ).
\end{align*}
Then our contention follows.
\end{proof}

\end{appendix}

\section*{Acknowledgement}
Lizhen Lin  thanks Professor Rabi Bhattacharya for useful discussions.
This work was supported by grant R01ES017240 from the National Institute of Environmental Health Sciences (NIEHS) of the National Institute of Health.

\end{document}